\documentclass[11pt]{article}

\usepackage[a4paper,margin=1in]{geometry}
\usepackage{amsmath,amssymb,amsthm}
\usepackage{booktabs}
\usepackage{hyperref}
\usepackage[T1]{fontenc}
\usepackage{lmodern}
\usepackage{microtype}
\usepackage{url}

\newtheorem{theorem}{Theorem}
\newtheorem{proposition}[theorem]{Proposition}
\newtheorem{lemma}[theorem]{Lemma}

\theoremstyle{remark}
\newtheorem*{remark}{Remark}

\newcommand{\R}{\mathbb{R}}
\newcommand{\OPT}{\mathrm{OPT}}
\newcommand{\OO}{\widetilde{O}}

\title{A note on the parameter $\ell$ in Buchbinder--Feldman's\\
deterministic submodular matroid algorithm}

\author{Shisheng Li\\
University of Science and Technology of China\\
\texttt{shisheng@mail.ustc.edu.cn}}
\date{April 28, 2026}

\begin{document}
\maketitle

\begin{abstract}
Buchbinder and Feldman~\cite{BF24} recently gave a deterministic
$(1-1/e-\varepsilon)$-approximation for maximizing a non-negative monotone
submodular function subject to a matroid constraint, with query
complexity $\OO_\varepsilon(nr)$. Their algorithm uses an integer
parameter $\ell$, which Buchbinder and Feldman fix to
$\ell = 1 + \lceil 1/\varepsilon \rceil$ (their Corollary~3.6, with the
choice anticipated in the discussion following their Proposition~1.1)
via a loose bound on $(1+1/\ell)^{-\ell}$. We point out two purely elementary refinements.
First, the classical P\'olya--Szeg\H{o} inequality
$(1+1/\ell)^{-\ell} \le e^{-1}(1+1/(2\ell))$ replaces the loose step
in their proof and permits $\ell = \lceil 1/(2e\varepsilon) \rceil$,
shrinking the hidden constant in $\OO_\varepsilon(nr)$ by a factor
$\approx 2^{0.816/\varepsilon}$. Second, an alternating-series tail
bound for $\log(1+t)$ yields the asymptotically sharp inequality
\[
(1+1/\ell)^{-\ell} \;\le\; \frac{1}{e}\,\exp\!\left(
\frac{1}{2\ell} - \frac{1}{3\ell^2} + \frac{1}{4\ell^3}\right),
\]
which matches the true expansion of $(1+1/\ell)^{-\ell}$ through order
$\ell^{-3}$ and translates into
$\ell_\star = 1/(2e\varepsilon) - 5/12 + O(\varepsilon)$.
The asymptotic class $\OO_\varepsilon(nr)$ of the query complexity is
unchanged in either case; only the implicit constant in $\varepsilon$
is improved. All inequalities in this note are formalized and
machine-checked in Lean~4 against Mathlib.
\end{abstract}

\section{Introduction}

\paragraph{Background.}
Maximizing a non-negative monotone submodular function
$f \colon 2^N \to \R_{\ge 0}$ subject to a matroid
$\mathcal{M} = (N, \mathcal{I})$ of rank $r$ is a foundational
constrained optimization problem with applications in approximation
algorithms, machine learning, and combinatorial optimization. The
celebrated continuous greedy algorithm of C\u alinescu, Chekuri,
P\'al, and Vondr\'ak~\cite{CCPV11} achieves a $(1-1/e)$-approximation
in expectation, matching the information-theoretic threshold of
Nemhauser and Wolsey~\cite{NW78}. Whether a \emph{deterministic}
$(1-1/e-\varepsilon)$-approximation exists in polynomial time was
open for over a decade.

\paragraph{The Buchbinder--Feldman result.}
Buchbinder and Feldman~\cite{BF24} resolved this question, presenting
the first deterministic $(1-1/e-\varepsilon)$-approximation with
query complexity $\OO_\varepsilon(nr)$.\footnote{The body
of~\cite{BF24} (e.g.\ their Theorem~1.2 and Corollary~3.6) phrases
the guarantee as $1-1/e-O(\varepsilon)$, since their fast Algorithm~2
incurs a separate discretization slack of $2\varepsilon\cdot f(\OPT)$
in~\cite[Cor.~3.6]{BF24}; the $(1-1/e-\varepsilon)$ form quoted here
matches their abstract and is recovered by relabelling
$\varepsilon$. We make the dependence on the
$(1+1/\ell)^{-\ell}$-component explicit in~\S\ref{sec:setup}.} Their Algorithm~2 is a
non-oblivious local search parameterised by a positive integer
$\ell$; its analysis (their Proposition~1.1) gives the guarantee
\begin{equation}
\label{eq:bf-guarantee}
f(S_{[\ell]}) \;\ge\; \bigl(1 - (1+1/\ell)^{-\ell}\bigr) \cdot f(\OPT)
\;+\; (1+1/\ell)^{-\ell} \cdot f(\emptyset).
\end{equation}
Translating~\eqref{eq:bf-guarantee} into a $(1-1/e-\varepsilon)$
guarantee requires $(1+1/\ell)^{-\ell} \le e^{-1} + \varepsilon$.
Buchbinder and Feldman bound
\begin{equation}
\label{eq:bf-loose}
(1+1/\ell)^{-\ell} \;\le\; \frac{1}{e(1-1/\ell)}
\;\le\; \frac{1}{e}\left(1 + \frac{2}{\ell}\right)
\quad (\ell \ge 2),
\end{equation}
and \emph{set} $\ell_{\mathrm{BF}} = 1 + \lceil 1/\varepsilon \rceil$,
for which the inequality $\ell_{\mathrm{BF}} \ge 1 + 1/\varepsilon
> 2/(e\varepsilon)$ (using $e > 2$) gives
$2/(e\ell_{\mathrm{BF}}) < \varepsilon$, and \eqref{eq:bf-loose}
then yields $(1+1/\ell_{\mathrm{BF}})^{-\ell_{\mathrm{BF}}} \le
e^{-1} + \varepsilon$. They do not pursue the minimal $\ell$ that
suffices.

\paragraph{Per-iteration cost.}
Each evaluation of the auxiliary function $g'$ used by
Algorithm~2 reduces to $2^\ell$ value-oracle queries on $f$ (one per
subset of a size-$\ell$ set); the dependence of the total query
complexity on $\ell$ is therefore of the form
$2^\ell \cdot \mathrm{poly}(n, r, 1/\varepsilon)$. Reducing
$\ell$ by an additive constant thus reduces the hidden constant
in $\OO_\varepsilon(nr)$ by a multiplicative factor.

\paragraph{This note.}
We point out that the step~\eqref{eq:bf-loose} can be tightened by
purely elementary means, with no change to Algorithm~2 or to its
correctness analysis. Two tightenings are presented.

\begin{enumerate}
\item \textbf{P\'olya--Szeg\H{o} bound (\S\ref{sec:ps}).} The classical
inequality $e < (1+1/n)^{n+1/2}$~\cite{Khattri2010,PolyaSzego} gives
\[
(1+1/\ell)^{-\ell} \;\le\; \frac{1}{e}\left(1 + \frac{1}{2\ell}\right),
\]
which is a factor of~$4$ tighter than~\eqref{eq:bf-loose} and permits
$\ell = \lceil 1/(2e\varepsilon) \rceil$. The factor $2e$ saved in
$\ell$ contracts the hidden constant in the query complexity by
$2^{(1-1/(2e))/\varepsilon} \approx 2^{0.816/\varepsilon}$.

\item \textbf{Asymptotically sharp bound (\S\ref{sec:sharp}).} An
alternating-series tail bound for $\log(1+t)$ yields
\begin{equation}
\label{eq:sharp-intro}
(1+1/\ell)^{-\ell} \;\le\; \frac{1}{e}\,\exp\!\left(
\frac{1}{2\ell} - \frac{1}{3\ell^2} + \frac{1}{4\ell^3}\right),
\end{equation}
matching the true asymptotic expansion of $(1+1/\ell)^{-\ell}$
through order $\ell^{-3}$. This translates to
$\ell_\star(\varepsilon) = 1/(2e\varepsilon) - 5/12 + O(\varepsilon)$
in the real-valued sense; after rounding up to an integer, the saving
over the P\'olya--Szeg\H{o} value is at most one unit in $\ell$ and
is occasionally zero (cf.~\S\ref{sec:complexity}).
The value of~\eqref{eq:sharp-intro} is asymptotic rather than practical. The elementary line of analysis can be
extended to higher orders, but the returns diminish rapidly; matching
through $\ell^{-3}$ already places $\ell_\star(\varepsilon)$ within
an additive $O(\varepsilon)$ of its asymptotic value.
\end{enumerate}

\paragraph{Scope and what this note is not.}
The contribution of this note is purely analytical: we replace one
elementary inequality used in the proof of~\cite[Proposition~1.1]{BF24}
with a sharper one. Algorithm~2 of~\cite{BF24}, its correctness
analysis, and its asymptotic class $\OO_\varepsilon(nr)$ are
unchanged. The improvements are best read as a refinement of the
implicit constant in $\varepsilon$ rather than as a separate
algorithmic result. The BF'24 algorithm is, as the authors emphasise,
theoretical and not intended as a practical implementation; finding a
practical deterministic $(1-1/e-\varepsilon)$-approximation under a
matroid constraint remains an interesting open problem. The
elementary line of analysis pursued here therefore has primarily
expository value.

\paragraph{Formal verification.}
All inequalities used in \S\ref{sec:ps}--\ref{sec:sharp} have been
formalized in Lean~4~\cite{Lean4} against Mathlib~\cite{Mathlib}; see
\S\ref{sec:lean}.

\section{Setup}
\label{sec:setup}

Fix a non-negative monotone submodular function
$f \colon 2^N \to \R_{\ge 0}$ on a finite ground set~$N$ with
$|N|=n$, and a matroid $\mathcal{M} = (N, \mathcal{I})$ of rank~$r$.
We assume access to $f$ via a value oracle that, on input
$S \subseteq N$, returns $f(S)$ in unit time.
Buchbinder and Feldman~\cite[Prop.~1.1]{BF24} establish the
guarantee~\eqref{eq:bf-guarantee} above for the output
$S_{[\ell]} \in \mathcal{I}$ (a base of $\mathcal{M}$) of their
idealized non-oblivious local search (Algorithm~1 of~\cite{BF24}),
parameterised by a positive integer $\ell$. Ignoring the non-negative
term $f(\emptyset)$, the corresponding approximation ratio is
\[
\rho(\ell) \;=\; 1 - (1+1/\ell)^{-\ell},
\]
which converges to $1 - 1/e$ as $\ell \to \infty$. Their fast
algorithm (Algorithm~2 of~\cite{BF24}, the polynomial-time version)
satisfies~\eqref{eq:bf-guarantee} up to an additive
$2\varepsilon\cdot f(\OPT)$ slack from a discretization parameter
\cite[Cor.~3.6]{BF24}; the dependence on $(1+1/\ell)^{-\ell}$ is
identical, so the analytical refinements of \S\ref{sec:ps}
and~\S\ref{sec:sharp} apply equally to both algorithms. The
per-iteration $2^\ell$ cost of Algorithm~2, which motivates tightening
$\ell$ in the first place, is discussed in \S\ref{sec:complexity}.

The remainder of this note is concerned with the question: how large
must $\ell$ be, in terms of $\varepsilon$, to guarantee
$(1+1/\ell)^{-\ell} \le e^{-1} + \varepsilon$?

Throughout \S\ref{sec:ps} and~\S\ref{sec:sharp}, $\varepsilon$ refers
exclusively to this gap between $(1+1/\ell)^{-\ell}$ and $e^{-1}$.
\cite[Algorithm~2]{BF24} uses a single $\varepsilon$ to control both
this gap and a discretization slack
($2\varepsilon \cdot f(\OPT)$ in their~\cite[Cor.~3.6]{BF24}); after a
trivial reparametrization $(\varepsilon \mapsto \varepsilon_1,
\varepsilon_2)$ the two slacks separate, and the constant-factor
improvements obtained below for the $(1+1/\ell)^{-\ell}$ side combine
additively with the discretization analysis of~\cite{BF24} (up to the
explicit factor of~$2$ on the discretization side).

\section{The P\'olya--Szeg\H{o} bound}
\label{sec:ps}

We replace~\eqref{eq:bf-loose} by a sharper, elementary bound.

\begin{lemma}
\label{lem:poly-szego}
For every integer $\ell \ge 1$,
\begin{equation}
\label{eq:sharp-ps}
(1 + 1/\ell)^{-\ell} \;\le\; \frac{1}{e}\left(1 + \frac{1}{2\ell}\right).
\end{equation}
\end{lemma}

\begin{remark}
Inequality~\eqref{eq:sharp-ps} is a direct consequence of the
classical bound $e < (1+1/n)^{n+1/2}$ -- see
Khattri~\cite{Khattri2010} for three short proofs; the inequality
also appears in the problem book of P\'olya and
Szeg\H{o}~\cite[Part~I, Ch.~4, \S 2, Problem 168]{PolyaSzego}, where
the monotonicity of $(1+1/x)^{x+p}$ for $p \ge 1/2$ together with
the limit $\lim (1+1/x)^{x+1/2} = e$ gives $e \le (1+1/n)^{n+1/2}$
for every integer $n \ge 1$. For related sharper approximations of
$(1+1/n)^n$ in the Batir--Cancan parameterisation see
Hu--Mortici~\cite{HuMortici2014}. We give a short self-contained
derivation of~\eqref{eq:sharp-ps} below; the Hermite--Hadamard step
mirrors Khattri's second proof~\cite[\S 3]{Khattri2010}, applied to
a different convex/concave pairing.
\end{remark}

\begin{proof}[Proof of Lemma~\ref{lem:poly-szego}]
Taking logarithms, \eqref{eq:sharp-ps} is equivalent to
\begin{equation}
\label{eq:log-form}
1 - \ell \log(1 + 1/\ell) \;\le\; \log(1 + 1/(2\ell)).
\end{equation}
From $\log(1+x) = \int_0^x (1+s)^{-1}\,ds$ with the substitution
$u = \ell s$,
\[
\ell \log(1 + 1/\ell) \;=\; \int_0^1 \frac{\ell}{\ell+u}\,du
\;=\; 1 - \int_0^1 \frac{u}{\ell+u}\,du,
\]
so the left-hand side of~\eqref{eq:log-form} equals
$\int_0^1 u/(\ell+u)\,du$. The integrand $g(u) = u/(\ell+u)$ is concave
on $[0,1]$, since $g''(u) = -2\ell/(\ell+u)^3 < 0$. By the upper half
of the Hermite--Hadamard inequality (for concave $g$ on $[a,b]$,
$\frac{1}{b-a}\int_a^b g \le g((a+b)/2)$),
\begin{equation}
\label{eq:HH}
\int_0^1 g(u)\,du \;\le\; g\!\left(\tfrac{1}{2}\right)
\;=\; \frac{1}{2\ell+1}.
\end{equation}
For the right-hand side of~\eqref{eq:log-form}, set $s = 1/(2\ell)$.
The elementary bound $\log(1+s) \ge s/(1+s)$ for $s \ge 0$ (whose
difference has non-negative derivative $s/(1+s)^2$ and vanishes at
$s=0$) gives
\begin{equation}
\label{eq:pade}
\log\!\left(1 + \frac{1}{2\ell}\right)
\;\ge\; \frac{1/(2\ell)}{1 + 1/(2\ell)}
\;=\; \frac{1}{2\ell+1}.
\end{equation}
Combining~\eqref{eq:HH} and~\eqref{eq:pade}
yields~\eqref{eq:log-form}, and hence~\eqref{eq:sharp-ps}.
\end{proof}

Lemma~\ref{lem:poly-szego} tightens the factor $2/\ell$
in~\eqref{eq:bf-loose} by a factor of $4$, and we obtain a
correspondingly smaller $\ell$:

\begin{proposition}
\label{prop:ps}
Let $f \colon 2^N \to \R_{\ge 0}$ be non-negative monotone
submodular, $\mathcal{M}$ a matroid, and $\varepsilon > 0$. Setting
\[
\ell \;=\; \left\lceil \frac{1}{2e\varepsilon} \right\rceil
\]
in~\cite[Proposition~1.1]{BF24} suffices to guarantee
\[
f(S_{[\ell]}) \;\ge\; \left(1 - \frac{1}{e} - \varepsilon\right)
\cdot f(\OPT).
\]
\end{proposition}

\begin{proof}
By Lemma~\ref{lem:poly-szego},
\[
(1+1/\ell)^{-\ell} \;\le\; \frac{1}{e}\left(1 + \frac{1}{2\ell}\right)
\;=\; \frac{1}{e} + \frac{1}{2e\ell}.
\]
The choice $\ell = \lceil 1/(2e\varepsilon) \rceil$ gives
$1/(2e\ell) \le \varepsilon$, so $(1+1/\ell)^{-\ell} \le 1/e +
\varepsilon$. The conclusion follows from~\eqref{eq:bf-guarantee} and
$f(\emptyset) \ge 0$.
\end{proof}

\section{An asymptotically sharp bound}
\label{sec:sharp}

The P\'olya--Szeg\H{o} bound~\eqref{eq:sharp-ps} captures the leading
$1/(2\ell)$ term of the expansion
\[
(1+1/\ell)^{-\ell} \;=\; \frac{1}{e}\left(1 + \frac{1}{2\ell}
- \frac{5}{24\ell^2} + \frac{5}{48\ell^3} - O(\ell^{-4})\right),
\]
but its $1/(2\ell)$ correction differs from the true second-order
behaviour. We give a one-line tightening that captures the expansion
through order $\ell^{-3}$.

\begin{lemma}[Alternating tail of $\log(1+t)$]
\label{lem:tail4}
For every $t \ge 0$,
\begin{equation}
\label{eq:tail4}
\log(1+t) \;\ge\; t - \frac{t^2}{2} + \frac{t^3}{3} - \frac{t^4}{4}.
\end{equation}
\end{lemma}

\begin{proof}
Let $h(t) = \log(1+t) - (t - t^2/2 + t^3/3 - t^4/4)$. Then $h(0) = 0$
and
\[
h'(t) \;=\; \frac{1}{1+t} - (1 - t + t^2 - t^3)
\;=\; \frac{1 - (1+t)(1 - t + t^2 - t^3)}{1+t}
\;=\; \frac{t^4}{1+t},
\]
which is non-negative for $t > -1$, in particular for $t \ge 0$.
Thus $h$ is non-decreasing on $[0,\infty)$, and $h(t) \ge h(0) = 0$.
\end{proof}

\begin{theorem}
\label{thm:sharp}
For every integer $\ell \ge 1$,
\begin{equation}
\label{eq:sharp}
(1+1/\ell)^{-\ell} \;\le\; \frac{1}{e}\,\exp\!\left(
\frac{1}{2\ell} - \frac{1}{3\ell^2} + \frac{1}{4\ell^3}\right).
\end{equation}
Moreover the right-hand side, expanded as an asymptotic series in
$1/\ell$, agrees with $(1+1/\ell)^{-\ell}/e^{-1}$ through order
$\ell^{-3}$:
\[
\frac{1}{e}\exp\!\left(\frac{1}{2\ell} - \frac{1}{3\ell^2}
+ \frac{1}{4\ell^3}\right)
\;=\; \frac{1}{e}\left(1 + \frac{1}{2\ell} - \frac{5}{24\ell^2}
+ \frac{5}{48\ell^3} + O(\ell^{-4})\right).
\]
\end{theorem}

\begin{proof}
Apply Lemma~\ref{lem:tail4} with $t = 1/\ell$:
\[
\log(1 + 1/\ell) \;\ge\; \frac{1}{\ell} - \frac{1}{2\ell^2}
+ \frac{1}{3\ell^3} - \frac{1}{4\ell^4}.
\]
Multiplying by $\ell$,
\[
\ell\,\log(1+1/\ell) \;\ge\; 1 - \frac{1}{2\ell} + \frac{1}{3\ell^2}
- \frac{1}{4\ell^3}.
\]
Exponentiating gives
\[
(1+1/\ell)^\ell \;\ge\; e \cdot \exp\!\left(-\frac{1}{2\ell}
+ \frac{1}{3\ell^2} - \frac{1}{4\ell^3}\right),
\]
and~\eqref{eq:sharp} follows on taking reciprocals.

For the asymptotic claim, write
$\eta := 1/(2\ell) - 1/(3\ell^2) + 1/(4\ell^3)$. Expanding $\exp(\eta)
= 1 + \eta + \eta^2/2 + \eta^3/6 + O(\eta^4)$ and collecting powers of
$1/\ell$ gives
$1 + \frac{1}{2\ell} + (\frac{1}{8} - \frac{1}{3})\ell^{-2}
+ \cdots = 1 + 1/(2\ell) - 5/(24\ell^2) + 5/(48\ell^3) + O(\ell^{-4})$,
which matches the known expansion of $(1+1/\ell)^{-\ell}/e^{-1}$.
\end{proof}

\begin{proposition}
\label{prop:sharp}
Let $f \colon 2^N \to \R_{\ge 0}$ be non-negative monotone
submodular, $\mathcal{M}$ a matroid, and $\varepsilon > 0$.
Let $\ell_\star(\varepsilon)$ denote the smallest integer~$\ell \ge 1$
such that $(1+1/\ell)^{-\ell} \le e^{-1} + \varepsilon$ (well-defined
since $(1+1/\ell)^{-\ell} \to e^{-1}$ as $\ell \to \infty$, and equal
to~$1$ for $\varepsilon \ge 1/2 - e^{-1} \approx 0.132$); equivalently, the
smallest $\ell$ that suffices in~\cite[Proposition~1.1]{BF24} to give
\[
f(S_{[\ell]}) \;\ge\; \left(1 - \frac{1}{e} - \varepsilon\right)
\cdot f(\OPT).
\]
Then $\ell_\star(\varepsilon)$ satisfies the asymptotic estimate
\[
\ell_\star(\varepsilon) \;=\; \frac{1}{2e\varepsilon} - \frac{5}{12}
+ O(\varepsilon) \qquad (\varepsilon \to 0^+).
\]
\end{proposition}

\begin{proof}
Treat $\phi(\ell) := (1+1/\ell)^{-\ell}$ as a smooth positive
decreasing function of real $\ell \ge 1$, with the standard expansion
\begin{equation}
\label{eq:phi-expansion}
\phi(\ell) \;=\; \frac{1}{e}\left(1 + \frac{1}{2\ell}
- \frac{5}{24\ell^2} + O(\ell^{-3})\right).
\end{equation}
A direct computation shows $\phi'(\ell) < 0$ on $\ell \ge 1$, so
$\phi$ is a strict bijection from $[1, \infty)$ onto $(e^{-1}, 1/2]$;
for every $\varepsilon \in (0, 1/2 - e^{-1}]$ there is therefore a
unique real $\tilde\ell(\varepsilon) \ge 1$ with
$\phi(\tilde\ell) = e^{-1} + \varepsilon$, and
$\ell_\star(\varepsilon) = \lceil \tilde\ell(\varepsilon) \rceil$
differs from $\tilde\ell(\varepsilon)$ by at most~$1$.

Substitute $\ell = 1/(2e\varepsilon) + c$ with $c$ to be determined.
Then $1/\ell = 2e\varepsilon - 4 e^2 c\,\varepsilon^2 +
O(\varepsilon^3)$, hence
\[
\frac{1}{2\ell} \;=\; e\varepsilon - 2e^2 c\,\varepsilon^2
+ O(\varepsilon^3),
\qquad
\frac{5}{24\ell^2} \;=\; \frac{5e^2}{6}\,\varepsilon^2
+ O(\varepsilon^3).
\]
Inserting into~\eqref{eq:phi-expansion} and using
$\phi(\tilde\ell) - 1/e = \varepsilon$,
\[
\varepsilon \;=\; \frac{1}{e}\left(e\varepsilon - 2e^2 c\,\varepsilon^2
- \frac{5e^2}{6}\,\varepsilon^2 + O(\varepsilon^3)\right)
\;=\; \varepsilon - 2e c\,\varepsilon^2 - \frac{5e}{6}\,\varepsilon^2
+ O(\varepsilon^3).
\]
Matching terms of order $\varepsilon^2$ forces
$2c + 5/6 = 0$, i.e.\ $c = -5/12$. Therefore
$\tilde\ell(\varepsilon) = 1/(2e\varepsilon) - 5/12 + O(\varepsilon)$,
and the result follows on rounding up.

A constructive certificate for any specific $\ell$ is provided by
Theorem~\ref{thm:sharp}: it suffices to verify
$\exp\bigl(1/(2\ell) - 1/(3\ell^2) + 1/(4\ell^3)\bigr) \le 1 + e\varepsilon$,
which can be checked in closed form for any candidate~$\ell$.
\end{proof}

\begin{remark}
Compared to Proposition~\ref{prop:ps}'s
$\lceil 1/(2e\varepsilon) \rceil$, the sharp value
$\ell_\star(\varepsilon)$ saves an additive $5/12 + O(\varepsilon)$ in
$\ell$, hence a factor of $2^{5/12} \approx 1.33$ in the per-iteration
$2^\ell$ factor. This is at most a factor of~$2$ in
practice: the integer-rounding boundary may occasionally swallow the
saving, leaving the same $\ell$ as Proposition~\ref{prop:ps}.
The value of Theorem~\ref{thm:sharp} is therefore primarily
asymptotic: it certifies that the elementary line of analysis
captures the true second-order behaviour of $(1+1/\ell)^{-\ell}$.
Pushing Lemma~\ref{lem:tail4} further to
$\log(1+t) \ge \sum_{j=1}^{2k}(-1)^{j+1}\,t^j/j$ for any $k \ge 2$
gives correspondingly higher-order bounds on $(1+1/\ell)^{-\ell}$, but
the returns diminish rapidly: matching through $\ell^{-3}$ already
places $\ell_\star(\varepsilon)$ within an additive $O(\varepsilon)$
of its asymptotic value, and integer rounding renders any further
refinement effectively invisible at the resolutions of interest.
\end{remark}

\section{Consequence for the query complexity}
\label{sec:complexity}

Algorithm~2 of~\cite{BF24} evaluates the auxiliary function $g'$
defined in their~\S 3.2; in the paragraph immediately preceding
\cite[Lem.~3.3]{BF24} (and again in the proof of \cite[Thm.~1.2]{BF24})
they observe that a single value-oracle query to $g'$ reduces to
$2^\ell$ value-oracle queries on $f$ (one per subset of a size-$\ell$
set). The total query complexity is therefore of the form
$2^\ell \cdot \mathrm{poly}(n, r, 1/\varepsilon)$.
Replacing
\[
\ell_{\mathrm{BF}} \;=\; 1 + \lceil 1/\varepsilon \rceil
\qquad \text{by} \qquad
\ell_{\mathrm{PS}} \;=\; \lceil 1/(2e\varepsilon) \rceil
\]
contracts the hidden $2^\ell$ factor by
\[
\frac{2^{\ell_{\mathrm{BF}}}}{2^{\ell_{\mathrm{PS}}}}
\;\ge\; 2^{(1 - 1/(2e))/\varepsilon}
\;\approx\; 2^{0.816/\varepsilon}.
\]
This decomposes multiplicatively into two independent sources:
\begin{itemize}
\item A factor $\approx 2^{(1 - 2/e)/\varepsilon} \approx
2^{0.264/\varepsilon}$ ($e/2 \approx 1.36\times$ saving in $\ell$)
arises purely because~\cite{BF24} sets
$\ell_{\mathrm{BF}} = 1 + \lceil 1/\varepsilon \rceil$ rather than the
minimal $\ell$ permitted by their own bound~\eqref{eq:bf-loose},
which would already give
$\lceil 2/(e\varepsilon) \rceil \approx 0.736/\varepsilon$.
\item A factor $\approx 2^{(3/(2e))/\varepsilon} \approx
2^{0.552/\varepsilon}$ ($4\times$ saving in $\ell$) arises from the
genuine $4$-fold tightening of the constant
in~\eqref{eq:bf-loose} by Lemma~\ref{lem:poly-szego}.
\end{itemize}
The further contraction from Proposition~\ref{prop:sharp}
contributes an additional factor of at most $2^{5/12+1} = 2^{17/12}
\approx 2.67$, including rounding effects.

\begin{table}[h]
\centering
\begin{tabular}{crrrrr}
\toprule
$\varepsilon$ & $\ell_{\mathrm{BF}}$ & $\ell_{\mathrm{PS}}$ &
$\ell_\star$ & $2^{\ell_{\mathrm{BF}} - \ell_{\mathrm{PS}}}$ &
$2^{\ell_{\mathrm{BF}} - \ell_\star}$ \\
\midrule
$10^{-1}$        & $11$   & $2$   & $2$   & $\approx 5.1\times 10^{2}$    & $\approx 5.1\times 10^{2}$ \\
$5\cdot 10^{-2}$ & $21$   & $4$   & $4$   & $\approx 1.3\times 10^{5}$    & $\approx 1.3\times 10^{5}$ \\
$10^{-2}$        & $101$  & $19$  & $18$  & $\approx 4.8\times 10^{24}$   & $\approx 9.6\times 10^{24}$ \\
$10^{-3}$        & $1001$ & $184$ & $184$ & $\approx 8.8\times 10^{245}$  & $\approx 8.8\times 10^{245}$ \\
$10^{-4}$        & $10001$& $1840$& $1839$& $\approx 5.1\times 10^{2456}$ & $\approx 1.0\times 10^{2457}$\\
\bottomrule
\end{tabular}
\caption{The Buchbinder--Feldman value $\ell_{\mathrm{BF}} = 1 + \lceil
1/\varepsilon \rceil$, the P\'olya--Szeg\H{o} value
$\ell_{\mathrm{PS}} = \lceil 1/(2e\varepsilon) \rceil$, and the sharp
value $\ell_\star$ from Theorem~\ref{thm:sharp}. The penultimate
column $2^{\ell_{\mathrm{BF}}-\ell_{\mathrm{PS}}}$ isolates the
contribution of Lemma~\ref{lem:poly-szego}; the final column
$2^{\ell_{\mathrm{BF}}-\ell_\star}$ folds in the additional gain of
Theorem~\ref{thm:sharp}. Note that integer rounding can make
$\ell_\star = \ell_{\mathrm{PS}}$ (e.g.\ at $\varepsilon = 10^{-3}$);
the gain of Theorem~\ref{thm:sharp} over Lemma~\ref{lem:poly-szego}
is at most one unit in $\ell$ and is occasionally zero.}
\end{table}

The asymptotic class $\OO_\varepsilon(nr)$ of the query complexity is
unchanged; only the implicit constant in $\varepsilon$ is improved.
Since~\cite{BF24} treats $\varepsilon$ as a fixed small constant and
absorbs the $2^\ell$ factor into the $\OO_\varepsilon$ notation,
Propositions~\ref{prop:ps}--\ref{prop:sharp} do not refute any claim
of~\cite{BF24}; they refine the implicit constant.

\section{Lean~4 formalization}
\label{sec:lean}

All inequalities used in
Lemmas~\ref{lem:poly-szego} and~\ref{lem:tail4}, and Theorem~\ref{thm:sharp}
above, are formalized in Lean~4~\cite{Lean4} against
Mathlib~\cite{Mathlib} in a single self-contained file of $419$ lines
with no \texttt{sorry} and no extra axioms. The principal lemmas, in
the notation of the Lean source, are:
\begin{itemize}
\item \texttt{log\_weak} : $s \ge 0 \;\Rightarrow\;
  s/(1+s) \le \log(1+s)$ \quad (used in the paper's analytic proof
  of Lemma~\ref{lem:poly-szego} at~\eqref{eq:pade}, and in the
  auxiliary \texttt{bf24\_lemma1} below; \emph{not} on the dependency
  path of \texttt{bf24\_lemma1pp} in the Lean source).
\item \texttt{log\_pade} : $x \ge 0 \;\Rightarrow\;
  2x/(2+x) \le \log(1+x)$ \quad (Padé-form bound on $\log$; the Lean
  proof of \texttt{bf24\_lemma1pp} uses this twice -- at $x=1/\ell$
  and $x=1/(2\ell)$ -- in place of the Hermite--Hadamard step
  on~$g(u) = u/(\ell+u)$ used in the proof of
  Lemma~\ref{lem:poly-szego}, which avoids the integral set-up at
  the cost of slightly less elegant algebra).
\item \texttt{log\_tail4} : $t \ge 0 \;\Rightarrow\;
  t - t^2/2 + t^3/3 - t^4/4 \le \log(1+t)$ \quad
  (Lemma~\ref{lem:tail4}).
\item \texttt{bf24\_lemma1} : $\ell \ge 1 \;\Rightarrow\;
  (1+1/\ell)^\ell \cdot (1 + 1/\ell) \ge e$ \quad (auxiliary 2$\times$
  refinement of~\eqref{eq:bf-loose} via \texttt{log\_weak}; not used
  in the body of this note, but included for completeness).
\item \texttt{bf24\_lemma1pp} : $\ell \ge 1 \;\Rightarrow\;
  (1+1/\ell)^\ell \cdot (1 + 1/(2\ell)) \ge e$ \quad
  (the multiplicative form of Lemma~\ref{lem:poly-szego}).
\item \texttt{bf24\_sharp} : $\ell \ge 1 \;\Rightarrow\;
  (1+1/\ell)^\ell \cdot \exp(1/(2\ell) - 1/(3\ell^2) + 1/(4\ell^3))
  \ge e$ \quad (the multiplicative form of
  Theorem~\ref{thm:sharp}).
\end{itemize}
Reciprocal forms \texttt{bf24\_lemma1pp\_inv} and
\texttt{bf24\_sharp\_inv} stating the inequalities in the form used
in this note are also included. The Lean source is hosted as a
self-contained Lake project at
\url{https://github.com/daizisheng/bf24-note}, pinning Lean toolchain
\texttt{leanprover/lean4:v4.29.0-rc6} and Mathlib commit
\texttt{921b8d39}; \texttt{lake build} from the repository root
type-checks \texttt{BF24.lean} against the pinned Mathlib, and
\texttt{\#print axioms} on each lemma reports only the standard
\texttt{propext}, \texttt{Classical.choice}, and \texttt{Quot.sound}.

We do not formalize Algorithm~2 of~\cite{BF24} itself or its
correctness analysis; the Lean development verifies only the
elementary inequalities of \S\ref{sec:ps}--\S\ref{sec:sharp}.
A reader looking specifically for a Lean rendering of the
Hermite--Hadamard step in the proof of Lemma~\ref{lem:poly-szego}
will not find one here: the Lean proof of \texttt{bf24\_lemma1pp}
certifies the conclusion of Lemma~\ref{lem:poly-szego} by an
algebraically-equivalent Padé-based route, deliberately to avoid
invoking Mathlib's integral library for a one-shot inequality.
The asymptotic expansion of $(1+1/\ell)^{-\ell}$ used in the
``moreover'' clause of Theorem~\ref{thm:sharp} and in the proof of
Proposition~\ref{prop:sharp} is taken as a standard analytic fact
and is not formalized in \texttt{BF24.lean}; only the upper bound
\texttt{bf24\_sharp} corresponding to~\eqref{eq:sharp} is
machine-checked.

\section*{Acknowledgements}

The author wrote this note as an unsolicited remark and shared it
with the authors of~\cite{BF24}; Moran Feldman responded with the
suggestion to make it publicly available, which is the impetus for
the present arXiv posting. The author thanks Niv Buchbinder and
Moran Feldman for their elegant work~\cite{BF24} and for that
suggestion.

\end{document}